\documentclass{article}
\usepackage{graphicx,geometry,verbatim}
\usepackage{amssymb}
\usepackage{amsmath}
\usepackage{verbatim}
\usepackage{setspace}
\usepackage{fullpage}
\usepackage{float}

\floatstyle{ruled}
\newfloat{Algorithm}{thp}{lop}[section]

\newtheorem{lemma}{Lemma}
\newenvironment{proof}{\noindent\textbf{Proof}}{\hfill\qed}
\newtheorem{theorem}{Theorem}
\newtheorem{proposition}{Proposition}

\newtheorem{definition}{Definition}

\newcommand{\qed}{\hfill$\Box$}

\begin{document}

\title{
Bounding the Impact of Unbounded Attacks in Stabilization\protect\footnote{A preliminary version of this work appears in the proceedings of the 8th International Symposium on Stabilization, Safety, and Security of Distributed Systems (SSS'06), see \cite{MT06cb}.}}

\author{Swan Dubois\protect\footnote{LIP6, Universit\'e Pierre et Marie Curie \& INRIA, France, swan.dubois@lip6.fr} \and Toshimitsu Masuzawa\protect\footnote{Osaka University, Japan, masuzawa@ist.osaka-u.ac.jp} \and S\'{e}bastien Tixeuil\protect\footnote{LIP6, Universit\'e Pierre et Marie Curie \& INRIA, France, sebastien.tixeuil@lip6.fr}}

\date{}

\maketitle

\begin{abstract}
Self-stabilization is a versatile approach to fault-tolerance since it permits a distributed system to recover from any transient fault that arbitrarily corrupts the contents of all memories in the system. Byzantine tolerance is an attractive feature of distributed systems that permits to cope with arbitrary malicious behaviors. Combining these two properties proved difficult: it is impossible to contain the spatial impact of Byzantine nodes in a self-stabilizing context for global tasks such as tree orientation and tree construction.

We present and illustrate a new concept of Byzantine containment in stabilization. Our property, called \emph{Strong Stabilization} enables to contain the impact of Byzantine nodes if they actually perform too many Byzantine actions. We derive impossibility results for strong stabilization and present strongly stabilizing protocols for tree orientation and tree construction that are optimal with respect to the number of Byzantine nodes that can be tolerated in a self-stabilizing context.
\end{abstract}


\section{Introduction}

The advent of ubiquitous large-scale distributed systems advocates that tolerance to various kinds of faults and hazards must be included from the very early design of such systems. \emph{Self-stabilization}~\cite{D74j,D00b,T09bc} is a versatile technique that permits forward recovery from any kind of \emph{transient} faults, while \emph{Byzantine Fault-tolerance}~\cite{LSP82j} is traditionally used to mask the effect of a limited number of \emph{malicious} faults. Making distributed systems tolerant to both transient and malicious faults is appealing yet proved difficult~\cite{DW04j,DD05c,NA02c} as impossibility results are expected in many cases.

Two main paths have been followed to study the impact of Byzantine faults in the context of self-stabilization:
\begin{enumerate}
\item \emph{Byzantine fault masking.} In completely connected synchronous systems, one of the most studied problems in the context of self-stabilization with Byzantine faults is that of \emph{clock synchronization}. In~\cite{BDH08c,DW04j}, probabilistic self-stabilizing protocols were proposed for up to one third of Byzantine processors, while in \cite{DH07cb,HDD06c} deterministic solutions tolerate up to one fourth and one third of Byzantine processors, respectively.

\item \emph{Byzantine containment.} For \emph{local} tasks (\emph{i.e.} tasks whose correctness can be checked locally, such as vertex coloring, link coloring, or dining philosophers), the notion of \emph{strict} stabilization was proposed~\cite{NA02c,SOM05c,MT07j}. Strict stabilization guarantees that there exists a \emph{containment radius} outside which the effect of permanent faults is masked. In \cite{NA02c}, the authors show that this Byzantine containment scheme is possible only for \emph{local} tasks. As many problems are not local, it turns out that it is impossible to provide strict stabilization for those.
\end{enumerate}

\paragraph{Our Contribution}
In this paper, we investigate the possibility of Byzantine containment in a self-stabilizing setting for tasks that are global (\emph{i.e.} for with there exists a causality chain of size $r$, where $r$ depends on $n$ the size of the network), and focus on two global problems, namely tree orientation and tree construction. As strict stabilization is impossible with such global tasks, we weaken the containment constraint by limiting the number of times that correct processes can be disturbed by Byzantine ones. Recall that strict stabilization requires that processes beyond the containment radius eventually achieve their desired behavior and are never disturbed by Byzantine processes afterwards. We relax this requirement in the following sense: we allow these correct processes beyond the containment radius to be disturbed by Byzantine processes, but only a limited number of times, even if Byzantine nodes take an infinite number of actions.

The main contribution of this paper is to present new possibility results for containing the influence of unbounded Byzantine behaviors. In more details, we define the notion of \emph{strong stabilization} as the novel form of the containment and introduce \emph{disruption times} to quantify the quality of the containment. The notion of strong stabilization is weaker than the strict stabilization but is stronger than the classical notion of self-stabilization (\emph{i.e.} every strongly stabilizing protocol is self-stabilizing, but not necessarily strictly stabilizing). While strict stabilization aims at tolerating an unbounded number of Byzantine processes, we explicitly refer the number of Byzantine processes to be tolerated. A self-stabilizing protocol is $(t,c,f)$-strongly stabilizing if the subsystem consisting of processes more than $c$ hops away from any Byzantine process is disturbed at most $t$ times in a distributed system with at most $f$ Byzantine processes. Here $c$ denotes the containment radius and $t$ denotes the disruption time.

To demonstrate the possibility and effectiveness of our notion of  strong stabilization, we consider \emph{tree construction} and \emph{tree orientation}. It is shown in \cite{NA02c} that there exists no strictly stabilizing protocol with a constant containment radius for these problems. The impossibility result can be extended even when the number of Byzantine processes is upper bounded (by one). In this paper, we provide a $(f\Delta^d, 0, f)$-strongly stabilizing protocol for rooted tree construction, provided that correct processes remain connected, where $n$ (respectively $f$) is the number of processes (respectively Byzantine processes) and $d$ is the diameter of the subsystem consisting of all correct processes. The containment radius of $0$ is obviously optimal. We show that the problem of tree orientation has no constant bound for the containment radius in a tree with two Byzantine processes even when we allow processes beyond the containment radius to be disturbed a finite number of times. Then we consider the case of a single Byzantine process and present a $(\Delta,0,1)$-strongly stabilizing protocol for tree orientation, where $\Delta$ is the maximum degree of processes. The containment radius of $0$ is also optimal. Notice that each process does not need to know the number $f$ of Byzantine processes and that $f$ can be $n-1$ at the worst case.  In other words, the algorithm is adaptive in the sense that the disruption times depend on the actual number of Byzantine processes. Both algorithms are also optimal with respect to the number of tolerated Byzantine nodes.
 
\section{Preliminaries}\label{preliminaries}

\subsection{Distributed System}

A \emph{distributed system} $S=(P,L)$ consists of a set $P=\{v_1,v_2,\ldots,v_n\}$ of processes and a set $L$ of bidirectional communication links (simply called links). A link is an unordered pair of distinct processes. A distributed system $S$ can be regarded as a graph whose vertex set is $P$ and whose link set is $L$, so we use graph terminology to describe a distributed system $S$.

Processes $u$ and $v$ are called \emph{neighbors} if $(u,v)\in L$. The set of neighbors of a process $v$ is denoted by $N_v$, and its cardinality (the \emph{degree} of $v$) is denoted by $\Delta_v (=|N_v|)$. The degree $\Delta$ of a distributed system $S=(P,L)$ is defined as $\Delta = \max \{\Delta_v\ |\ v \in P\}$. We do not assume existence of a unique identifier for each process (that is, the system is anonymous). Instead we assume each process can distinguish its neighbors from each other by locally arranging them in some arbitrary order: the $k$-th neighbor of a process $v$ is denoted by $N_v(k)\ (1 \le k \le \Delta_v)$.

Processes can communicate with their neighbors through \emph{link registers}. For each pair of neighboring processes $u$ and $v$, there are two link registers $r_{u,v}$ and $r_{v,u}$. Message transmission from $u$ to $v$ is realized as follows: $u$ writes a message to link register $r_{u,v}$ and then $v$ reads it from $r_{u,v}$. The link register $r_{u,v}$ is called an \emph{output register} of $u$ and is called an \emph{input register} of $v$. The set of all output (respesctively input) registers of $u$ is denoted by $Out_u$ (respectively $In_u$), \emph{i.e.} $Out_u=\{r_{u,v}\ |\ v \in N_u\}$ and $In_u=\{r_{v,u}\ | v \in N_u\}$.

The variables that are maintained by processes denote process states. Similarly, the values of the variables stored in each link register denote the state of the registers. A process may take actions during the execution of the system. An action is simply a function that is executed in an atomic manner by the process. The actions executed by each process is described by a finite set of guarded actions of the form $\langle$guard$\rangle\longrightarrow\langle$statement$\rangle$. Each guard of process $u$ is a boolean expression involving the variables of $u$ and its input registers. Each statement of process $u$ is an update of its state and its output/input registers.

A global state of a distributed system is called a \emph{configuration} and is specified by a product of states of all processes and all link registers. We define $C$ to be the set of all possible configurations of a distributed system $S$. For a process set $R \subseteq P$ and two configurations $\rho$ and $\rho'$, we denote $\rho \stackrel{R}{\mapsto} \rho'$ when $\rho$ changes to $\rho'$ by executing an action of each process in $R$ simultaneously. Notice that $\rho$ and $\rho'$ can be different only in the states of processes in $R$ and the states of their output registers. For completeness of execution semantics, we should clarify the configuration resulting from simultaneous actions of neighboring processes. The action of a process depends only on its state at $\rho$ and the states of its input registers at $\rho$, and the result of the action reflects on the states of the process and its output registers at $\rho '$.

A \emph{schedule} of a distributed system is an infinite sequence of process sets.  Let $Q=R^1, R^2, \ldots$  be a schedule, where $R^i \subseteq P$ holds for each $i\ (i \ge 1)$. An infinite sequence of configurations $e=\rho_0,\rho_1,\ldots$ is called an \emph{execution} from an initial configuration $\rho_0$ by a schedule $Q$, if $e$ satisfies $\rho_{i-1} \stackrel{R^i}{\mapsto} \rho_i$ for each $i\ (i \ge 1)$. Process actions are executed atomically, and we also assume that a \emph{distributed daemon} schedules the actions of processes, \emph{i.e.} any subset of processes can simultaneously execute their actions. 

The set of all possible executions from $\rho_0\in C$ is denoted by $E_{\rho_0}$. The set of all possible executions is denoted by $E$, that is, $E=\bigcup_{\rho\in C}E_{\rho}$. We consider \emph{asynchronous} distributed systems where we can make no assumption on schedules except that any schedule is \emph{weakly fair}: every process is contained in infinite number of subsets appearing in any schedule.

In this paper, we consider (permanent) \emph{Byzantine faults}: a Byzantine process (\emph{i.e.} a Byzantine-faulty process) can make arbitrary behavior independently from its actions. If $v$ is a Byzantine process, $v$ can repeatedly change its variables and its out put registers arbitrarily.

In asynchronous distributed systems, time is usually measured by \emph{asynchronous rounds} (simply called \emph{rounds}). Let $e=\rho_0,\rho_1,\ldots$ be an execution by a schedule $Q=R^1,R^2,\ldots$. The first round of $e$ is defined to be the minimum prefix of $e$, $e'=\rho_0,\rho_1,\ldots,\rho_k$, such that $\bigcup_{i=1}^k R^i =P'$ where $P'$ is the set of correct processes of $P$. Round $t\ (t\ge 2)$ is defined recursively, by applying the above definition of the first round to $e''=\rho_k,\rho_{k+1},\ldots$. Intuitively, every correct process has a chance to update its state in every round.

\subsection{Self-Stabilizing Protocol Resilient to Byzantine Faults}

Problems considered in this paper are so-called \emph{static problems}, \emph{i.e.} they require the system to find static solutions. For example, the spanning-tree construction problem is a static problem, while the mutual exclusion problem is not. Some static problems can be defined by a \emph{specification predicate} (shortly, specification), $spec(v)$, for each process $v$: a configuration is a desired one (with a solution) if every process satisfies $spec(v)$. A specification $spec(v)$ is a boolean expression on variables of $P_v~(\subseteq P)$ where $P_v$ is the set of processes whose variables appear in $spec(v)$. The variables appearing in the specification are called \emph{output variables} (shortly, \emph{O-variables}). In what follows, we consider a static problem defined by specification $spec(v)$.

A \emph{self-stabilizing protocol} is a protocol that eventually reaches a \emph{legitimate configuration}, where $spec(v)$ holds at every process $v$, regardless of the initial configuration. Once it reaches a legitimate configuration, every process $v$ never changes its O-variables and always satisfies $spec(v)$. From this definition, a self-stabilizing protocol is expected to tolerate any number and any type of transient faults since it can eventually recover from any configuration affected by the transient faults. However, the recovery from any configuration is guaranteed only when every process correctly executes its action from the configuration, \emph{i.e.}, we do not consider existence of permanently faulty processes.

When (permanent) Byzantine processes exist, Byzantine processes may not satisfy $spec(v)$. In addition, correct processes near the Byzantine processes can be influenced and may be unable to satisfy $spec(v)$. Nesterenko and Arora~\cite{NA02c} define a \emph{strictly stabilizing protocol} as a self-stabilizing protocol resilient to unbounded number of Byzantine processes.

Given an integer $c$, a \emph{$c$-correct process} is a process defined as follows.

\begin{definition}[$c$-correct process] 
A process is $c$-correct if it is correct (\emph{i.e.} not Byzantine) and located at distance more than $c$ from any Byzantine process.
\end{definition}

\begin{definition}[$(c,f)$-containment]\label{def:cfcontained}
A configuration $\rho$ is \emph{$(c,f)$-contained} for specification $spec$ if, given at most $f$ Byzantine processes, in any execution starting from $\rho$, every $c$-correct process $v$ always satisfies $spec(v)$ and never changes its O-variables.
\end{definition}

The parameter $c$ of Definition~\ref{def:cfcontained} refers to the \emph{containment radius} defined in \cite{NA02c}. The parameter $f$ refers explicitly to the number of Byzantine processes, while \cite{NA02c} dealt with unbounded number of Byzantine faults (that is $f\in\{0\ldots n\}$).

\begin{definition}[$(c,f)$-strict stabilization]\label{def:cfstabilizing}
A protocol is \emph{$(c,f)$-strictly stabilizing} for specification $spec$ if, given at most $f$ Byzantine processes, any execution $e=\rho_0,\rho_1,\ldots$ contains a configuration $\rho_i$ that is $(c,f)$-contained for $spec$.
\end{definition}

An important limitation of the model of \cite{NA02c} is the notion of $r$-\emph{restrictive} specifications. Intuitively, a specification is $r$-restrictive if it prevents combinations of states that belong to two processes $u$ and $v$ that are at least $r$ hops away. An important consequence related to Byzantine tolerance is that the containment radius of protocols solving those specifications is at least $r$. For some problems, such as the spanning tree construction we consider in this paper, $r$ can not be bounded to a constant. We can show that there exists no $(o(n),1)$-strictly stabilizing protocol for the spanning tree construction.

To circumvent the impossibility result, we define a weaker notion than the strict stabilization. Here, the requirement to the containment radius is relaxed, \emph{i.e.} there may exist processes outside the containment radius that invalidate the specification predicate, due to Byzantine actions. However, the impact of Byzantine triggered action is limited in times: the set of Byzantine processes may only impact the subsystem consisting of processes outside the containment radius a bounded number of times, even if Byzantine processes execute an infinite number of actions.

From the states of $c$-correct processes, \emph{$c$-legitimate configurations} and \emph{$c$-stable configurations} are defined as follows.

\begin{definition}[$c$-legitimate configuration]
A configuration $\rho$ is $c$-legitimate for \emph{spec} if every $c$-correct process $v$ satisfies $spec(v)$.
\end{definition}

\begin{definition}[$c$-stable configuration]
A configuration $\rho$ is $c$-stable if every $c$-correct process never changes the values of its O-variables as long as Byzantine processes make no action.
\end{definition}

Roughly speaking, the aim of self-stabilization is to guarantee that a distributed system eventually reaches a $c$-legitimate and $c$-stable configuration. However, a self-stabilizing system can be disturbed by Byzantine processes after reaching a $c$-legitimate and $c$-stable configuration. The \emph{$c$-disruption} represents the period where $c$-correct processes are disturbed by Byzantine processes and is defined as follows 

\begin{definition}[$c$-disruption]
A portion of execution $e=\rho_0,\rho_1,\ldots,\rho_t$ ($t>1$) is a $c$-disruption if and only if the following holds:
\begin{enumerate}
\item $e$ is finite,
\item $e$ contains at least one action of a $c$-correct process for changing the value of an O-variable,
\item $\rho_0$ is $c$-legitimate for \emph{spec} and $c$-stable, and
\item $\rho_t$ is the first configuration after $\rho_0$ such that $\rho_t$ is $c$-legitimate for \emph{spec} and $c$-stable.
\end{enumerate}
\end{definition}

Now we can define a self-stabilizing protocol such that Byzantine processes may only impact the subsystem consisting of processes outside the containment radius a bounded number of times, even if Byzantine processes execute an infinite number of actions.

\begin{definition}[$(t,k,c,f)$-time contained configuration]
A configuration $\rho_0$ is $(t,k,c,f)$-time contained for \emph{spec} if given at most $f$ Byzantine processes, the following properties are satisfied:
\begin{enumerate}
\item $\rho_0$ is $c$-legitimate for \emph{spec} and $c$-stable,
\item every execution starting from $\rho_0$ contains a $c$-legitimate configuration for \emph{spec} after which the values of all the O-variables of $c$-correct processes remain unchanged (even when Byzantine processes make actions repeatedly and forever), 
\item every execution starting from $\rho_0$ contains at most $t$ $c$-disruptions, and 
\item every execution starting from $\rho_0$ contains at most $k$ actions of changing the values of O-variables for each $c$-correct process.
\end{enumerate}
\end{definition}

\begin{definition}[$(t,c,f)$-strongly stabilizing protocol]
A protocol $A$ is $(t,c,f)$-strongly stabilizing if and only if starting from any arbitrary configuration, every execution involving at most $f$ Byzantine processes contains a $(t,k,c,f)$-time contained configuration that is reached after at most $l$ rounds. Parameters $l$ and $k$ are respectively the $(t,c,f)$-stabilization time and the $(t,c,f)$-process-disruption time of $A$.
\end{definition}

Note that a $(t,k,c,f)$-time contained configuration is a $(c,f)$-contained configuration when $t=k=0$, and thus, $(t,k,c,f)$-time contained configuration is a generalization (relaxation) of a $(c,f)$-contained configuration. Thus, a strongly stabilizing protocol is weaker than a strictly stabilizing one (as processes outside the containment radius may take incorrect actions due to Byzantine influence). However, a strongly stabilizing protocol is stronger than a classical self-stabilizing one (that may never meet their specification in the presence of Byzantine processes).

The parameters $t$, $k$ and $c$ are introduced to quantify the strength of fault containment, we do not require each process to know the values of the parameters. Actually, the protocols proposed in this paper assume no knowledge on the parameters.

There exists some relationship between these parameters as the following proposition states:

\begin{proposition}
If a configuration is $(t,k,c,f)$-time contained for \emph{spec}, then $t\leq nk$.
\end{proposition}

\begin{proof}
Let $\rho_0$ be a $(t,k,c,f)$-time contained configuration for \emph{spec}. Assume that $t>nk$.

If there exists no execution $e=\rho_0,\rho_1,\ldots$ such that $e$ contains at least $nk+1$ $c$-disruptions, then $\rho_0$ is in fact a $(nk,k,c,f)$-time contained configuration for \emph{spec} (and hence, we have $t\leq nk$). This is contradictory. So, there exists an execution $e=\rho_0,\rho_1,\ldots$ such that $e$ contains at least $nk+1$ $c$-disruptions.

As any $c$-disruption contains at least one action of a $c$-correct process for changing the value of an O-variable by definition, we obtain that $e$ contains at least $nk+1$ actions of $c$-correct processes for changing the values of O-variables. There is at most $n$ $c$-correct processes. So, there exists at least one $c$-correct process which takes at least $k+1$ actions for changing the value of O-variables in $e$. This is contradictory with the fact that $\rho_0$ is a $(t,k,c,f)$-time contained configuration for \emph{spec}.
\end{proof}

\subsection{Discussion}

There exists an analogy between the respective powers of $(c,f)$-strict stabilization and $(t,c,f)$-strong stabilization for the one hand, and self-stabilization and pseudo-stabilization for the other hand.

A \emph{pseudo-stabilizing} protocol (defined in~\cite{BGM93j}) guarantees that every execution has a suffix that matches the specification, but it could never reach a legitimate configuration from which any possible execution matches the specification. In other words, a pseudo-stabilizing protocol can continue to behave satisfying the specification, but with having possibility of invalidating the specification in future. A particular schedule can prevent a pseudo-stabilizing protocol from reaching a legitimate configuration for arbitrarily long time, but cannot prevent it from executing its desired behavior (that is, a behavior satisfying the specification) for arbitrarily long time. Thus, a pseudo-stabilizing protocol is useful since desired behavior is eventually reached.

Similarly, every execution of a $(t,c,f)$-strongly stabilizing protocol has a suffix such that every $c$-correct process executes its desired behavior. But, for a $(t,c,f)$-strongly stabilizing protocol, there may exist executions such that the system never reach a configuration after which Byzantine processes never have the ability to disturb the $c$-correct processes: all the $c$-correct processes can continue to execute their desired behavior, but with having possibility that the system (resp. each process) could be disturbed at most $t$ (resp. $k$) times by Byzantine processes in future. A notable but subtle difference is that the invalidation of the specification is caused only by the effect of Byzantine processes in a $(t,c,f)$-strongly stabilizing protocol, while the invalidation can be caused by a scheduler in a pseudo-stabilizing protocol.

\section{Strongly-Stabilizing Spanning Tree Construction}

\subsection{Problem Definition}

In this section, we consider only distributed systems in which a given process $r$ is distinguished as the root of the tree.

For \emph{spanning tree construction}, each process $v$ has an O-variable $prnt_v$ to designate a neighbor as its parent. Since processes have no identifiers, $prnt_v$ actually stores $k~(\in \{1,~2,\ldots,~\Delta_v\})$ to designate its $k$-th neighbor as its parent. No neighbor is designated as the parent of $v$ when $prnt_v = 0$ holds. For simplicity, we use $prnt_v=k\ (\in \{1,~2,\ldots,~\Delta_v\}$) and $prnt_v=u$ (where $u$ is the $k$-th neighbor of $v\in N_v(k)$) interchangeably, and $prnt_v=0$ and $prnt_v=\bot$ interchangeably.

The goal of spanning tree construction is to set $prnt_v$ of every process $v$ to form a rooted spanning tree, where $prnt_r = 0$ should hold for the root process $r$.

We consider Byzantine processes that can behave arbitrarily. The faulty processes can behave as if they were any internal processes of the spanning tree, or even as if they were the root processes. The first restriction we make on Byzantine processes is that we assume the root process $r$ can start from an arbitrary state, but behaves correctly according to a protocol. Another restriction on Byzantine processes is that we assume that all the correct processes form a connected subsystem; Byzantine processes never partition the system.

It is impossible, for example, to distinguish the (real) root $r$ from the faulty processes behaving as the root, we have to allow that a spanning forest (consisting of multiple trees) is constructed, where each tree is rooted with a root, correct or faulty one.

We define the specification predicate $spec(v)$ of the tree construction as follows.
\[spec(v) : \begin{cases}
 (prnt_v = 0) \wedge (level_v = 0) \text{ if } v \text{ is the root } r \\
 (prnt_v \in \{1,\ldots,~\Delta_v\}) \wedge ((level_v = level_{prnt_v}+1)\vee(prnt_v \text{ is Byzantine})) \text{ otherwise}
\end{cases}\]

Notice that $spec(v)$ requires that a spanning tree is constructed at any $0$-legitimate configuration, when no Byzantine process exists.

Figure~\ref{fig:ST} shows an example of $0$-legitimate configuration with Byzantine processes. The arrow attached to each process points the neighbor designated as its parent.

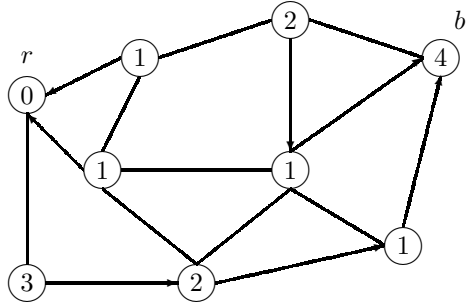
\begin{figure}[t]
\noindent \begin{centering} \ifx\JPicScale\undefined\def\JPicScale{0.5}\fi
\unitlength \JPicScale mm
\begin{picture}(120,80)(0,0)
\linethickness{0.3mm}
\put(5,55){\circle{10}}

\linethickness{0.3mm}
\put(35,65){\circle{10}}

\linethickness{0.3mm}
\put(75,75){\circle{10}}

\linethickness{0.3mm}
\put(75,35){\circle{10}}

\linethickness{0.3mm}
\put(25,35){\circle{10}}

\linethickness{0.3mm}
\put(115,65){\circle{10}}

\linethickness{0.3mm}
\put(105,15){\circle{10}}

\linethickness{0.3mm}
\put(50,5){\circle{10}}

\linethickness{0.3mm}
\put(5,5){\circle{10}}

\linethickness{0.3mm}
\multiput(75,40)(0.17,0.12){208}{\line(1,0){0.17}}
\put(110,65){\vector(4,3){0.12}}
\linethickness{0.3mm}
\multiput(105,20)(0.12,0.48){83}{\line(0,1){0.48}}
\put(115,60){\vector(1,4){0.12}}
\linethickness{0.3mm}
\multiput(10,55)(0.24,0.12){83}{\line(1,0){0.24}}
\put(10,55){\vector(-2,-1){0.12}}
\linethickness{0.3mm}
\multiput(5,50)(0.12,-0.12){125}{\line(1,0){0.12}}
\put(5,50){\vector(-1,1){0.12}}
\linethickness{0.3mm}
\put(5,10){\line(0,1){40}}
\linethickness{0.3mm}
\put(75,40){\line(0,1){30}}
\put(75,40){\vector(0,-1){0.12}}
\put(5,65){\makebox(0,0)[cc]{$r$}}

\put(120,75){\makebox(0,0)[cc]{$b$}}

\put(75,35){\makebox(0,0)[cc]{1}}

\put(105,15){\makebox(0,0)[cc]{1}}

\linethickness{0.3mm}
\multiput(40,65)(0.36,0.12){83}{\line(1,0){0.36}}
\linethickness{0.3mm}
\multiput(80,75)(0.36,-0.12){83}{\line(1,0){0.36}}
\linethickness{0.3mm}
\multiput(50,10)(0.15,0.12){167}{\line(1,0){0.15}}
\linethickness{0.3mm}
\multiput(75,30)(0.2,-0.12){125}{\line(1,0){0.2}}
\linethickness{0.3mm}
\multiput(55,5)(0.54,0.12){83}{\line(1,0){0.54}}
\put(100,15){\vector(4,1){0.12}}
\linethickness{0.3mm}
\put(30,35){\line(1,0){40}}
\linethickness{0.3mm}
\multiput(25,40)(0.12,0.24){83}{\line(0,1){0.24}}
\put(75,75){\makebox(0,0)[cc]{2}}

\put(35,65){\makebox(0,0)[cc]{1}}

\put(5,55){\makebox(0,0)[cc]{0}}

\put(25,35){\makebox(0,0)[cc]{1}}

\put(50,5){\makebox(0,0)[cc]{2}}

\put(115,65){\makebox(0,0)[cc]{4}}

\put(5,5){\makebox(0,0)[cc]{3}}

\linethickness{0.3mm}
\put(10,5){\line(1,0){35}}
\put(45,5){\vector(1,0){0.12}}
\linethickness{0.3mm}
\multiput(25,30)(0.15,-0.12){167}{\line(1,0){0.15}}
\end{picture}
  \par\end{centering}
 \caption{A legitimate configuration for spanning tree construction (numbers denote the level of processes). $r$ is the (real) root and $b$ is a Byzantine process which acts as a (fake) root.}
 \label{fig:ST}
\end{figure}

\subsection{Protocol $ss$-$ST$}

In many self-stabilizing tree construction protocols (see the survey of \cite{G03r}), each process checks locally the consistence of its $level$ variable with respect to the one of its neighbors. When it detects an inconsistency, it changes its $prnt$ variable in order to choose a ``better'' neighbor. The notion of ``better'' neighbor is based on the global desired property on the tree (\emph{e.g.} shortest path tree, minimun spanning tree...).

When the system may contain Byzantine processes, they may disturb their neighbors by providing alternatively ``better'' and ``worse'' states. The key idea of protocol $ss$-$ST$ to circumvent this kind of perturbation is the following: when a correct process detects a local inconsistency, it does not choose a ``better'' neighbor but it chooses another neighbor according to a round robin order (along the set of its neighbor).

Figure \ref{fig:ssST} presents our strongly-stabilizing spanning tree construction protocol $ss$-$ST$ that can tolerate any number of Byzantine processes other than the root process (providing that the subset of correct processes remains connected). These assumptions are necessary since a Byzantine root or a set of Byzantine processes that disconnects the set of correct processes may disturb all the tree infinitely often. Then, it is impossible to provide a $(t,k,f)$-strongly stabilizing protocol for any finite integer $t$.

The protocol is composed of three rules. Only the root can execute the first one ({\tt GA0}). This rule sets the root in a legitimate state if it is not the case. Non-root processes may execute the two other rules ({\tt GA1} and {\tt GA2}). The rule {\tt GA1} is executed when the state of a process is not legitimate. Its execution leads the process to choose a new parent and to compute its local state in function of this new parent. The last rule ({\tt GA2}) is enabled when a process is in a legitimate state but there exists an inconsistence between its variables and its shared registers. The execution of this rule leads the process to compute the consistent values for all its shared registers.
  
\begin{figure}
\begin{center}
\begin{tabbing}
xxx \= xxx \= xxx \= xxx \= xxx \= xxx \= xxx \= xxx \= xx \=\kill
{\tt constants of process $v$} \\
\> $\Delta_v =$ the degree of $v$; \\
\> $N_v = $ the set of neighbors of $v$;  \\
{\tt variables of process $v$} \\
\> $prnt_v\in \{0,1,2,\ldots,\Delta_v\}$: integer; // $prnt_v=0$ if $v$ has no parent,\\ 
\> \> \> \> \> \> \> \> \> // $prnt_v=k\in \{1,2,\ldots,\Delta_v\}$ if $N_v[k]$ is the parent of $v$.    \\
\> $level_v$: integer; // distance from the root.\\
{\tt variables in shared register $r_{v,u}$} \\
\> $r$-$prnt_{v,u}$: boolean; // $r$-$prnt_{v,u}=$\textit{true} iff $u$ is a parent of $v$.    \\
\> $r$-$level_{v,u}$: integer; // the value of $level_v$ \\
{\tt predicates} \\
\> $pred_0: prnt_v \neq 0$ {\tt or} $level_v \neq 0$ {\tt or} $\exists w\in N_v,[(r$-$prnt_{v,w}, r$-$level_{v,w})\neq($\textit{false}$, 0)]$\\
\> $pred_1: prnt_v \notin \{1,~2,\ldots,~\Delta_v\}$ {\tt or} $level_v \ne r$-$level_{prnt_v,v}+1$ \\
\> $pred_2: (r$-$prnt_{v,prnt_v}, r$-$level_{v,prnt_v})\neq($\textit{true}$, level_v)$\\
\> \> \> {\tt or} $\exists w\in N_v-\{prnt_v\},[(r$-$prnt_{v,w}, r$-$level_{v,w})\neq($\textit{false}$, level_v)]$ \\
{\tt atomic action of the root $v=r$} // represented in form of guarded action \\
\> {\tt GA0:}$pred_0$ $\longrightarrow$ \\
\> \> \> $prnt_v :=0 ;$\\
\> \> \> $level_v := 0;$\\
\> \> \> {\tt for each} $w \in N_v$ {\tt do} $(r$-$prnt_{v,w}, r$-$level_{v,w}):=($\textit{false}$, 0)$;\\
{\tt atomic actions of $v\neq r$} // represented in form of guarded actions \\
\> {\tt GA1:}$pred_1 \longrightarrow$ \\
\> \> \> $prnt_v := next_v(prnt_v)$ {\tt where} $next_v(k)=(k$ {\tt mod} $\Delta_v)+1$;\\
\> \> \> $level_v := r$-$level_{prnt_v,v}+1;$\\
\> \> \> $(r$-$prnt_{v,prnt_v}, r$-$level_{v,prnt_v}):=($\textit{true}$, level_v)$;\\
\> \> \> {\tt for each} $w \in N_v-\{prnt_v\}$ {\tt do} $(r$-$prnt_{v,w}, r$-$level_{v,w}):=($\textit{false}$, level_v)$;\\
\> {\tt GA2:}$\neg pred_1$ {\tt and} $pred_2 \longrightarrow$ \\
\> \> \> $(r$-$prnt_{v,prnt_v}, r$-$level_{v,prnt_v}):=($\textit{true}$, level_v)$;\\
\> \> \> {\tt for each} $w \in N_v-\{prnt_v\}$ {\tt do} $(r$-$prnt_{v,w}, r$-$level_{v,w}):=($\textit{false}$, level_v)$;\\
\end{tabbing}
\end{center}
\caption{Protocol $ss$-$ST$ (actions of process $v$)}
\label{fig:ssST}
\end{figure}

\subsection{Proof of Strong Stabilization of $ss$-$ST$}

We cannot make any assumption on the initial values of register variables. But, we can observe that if an output register of a correct process has inconsistent values with the process variables then this process is enabled by a rule of $ss$-$ST$. By fairness assumption, any such process takes a step in a finite time.

Once a correct process $v$ executes one of its action, variables of its output registers have values consistent with the process variables: $r$-$prnt_{v,prnt_v}=true$, $r$-$prnt_{v,w}=false\ (w \in N_v-\{prnt_v\})$, and $r$-$level_{v,w}=level_v\ (w \in N_v)$ hold.

Consequently, we can assume in the following that all the variables of output registers of every correct process have consistent values with the process variables.

We denote by $\cal{LC}$ the following set of configurations:
\[\begin{array}{rcl}
\cal{LC}&=&\Big\{\rho\in C \Big| (prnt_r=0) \wedge (level_r=0) \wedge\\
&& ~~~~~~~~~~\big(\forall v\in V-(B\cup\{r\}),(prnt_v \in \{1,\ldots,~\Delta_v\}) \wedge (level_v = level_{prnt_v}+1)\big)\Big\}
\end{array}\]

We interest now on properties of configurations of $\cal{LC}$.

\begin{lemma}\label{lem:closure}
Any configuration of $\cal{LC}$ is $0$-legitimate and $0$-stable.
\end{lemma}

\begin{proof}
Let $\rho$ be a configuration of $\cal{LC}$. By definition of $spec$, it is obvious that $\rho$ is $0$-legitimate.

Note that no correct process is enabled by $ss$-$ST$ in $\rho$. Consequently, no actions of $ss$-$ST$ can be executed and we can deduce that $\rho$ is $0$-stable.
\end{proof}

We can observe that there exists some $0$-legitimate configurations which not belong to $\cal{LC}$ (for example the one of Figure \ref{fig:ssST}).

\begin{lemma}\label{lem:convergence}
Given at most $n-1$ Byzantine processes, for any initial configuration $\rho_0$ and any execution $e=\rho_0,\rho_1,\ldots$ starting from $\rho_0$, there exists a configuration $\rho_i$ such that $\rho_i\in\cal{LC}$.
\end{lemma}

\begin{proof}
First, note that if all the correct processes are disabled in a configuration $\rho$, then $\rho$ belongs to $\cal{LC}$. Thus, it is sufficient to show that $ss$-$ST$ eventually reaches a configuration $\rho_i$ in any execution (starting from any configuration) such that all the correct processes are disabled in $\rho_i$.

By contradiction, assume that there exists a correct process that is enabled infinitely often. Notice that once the root process $r$ is activated, $r$ becomes and remains disabled forever. From the assumption that all the correct processes form a connected subsystem, there exists two neighboring correct processes $u$ and $v$ such that $u$ becomes and remains disabled and $v$ is enabled infinitely often. Consider execution after $u$ becomes and remains disabled. Since the daemon is weakly fair, $v$ executes its action infinitely often. Then, eventually $v$ designates $u$ as its parent. It follows that $v$ never becomes enabled again unless $u$ changes $level_u$. Since $u$ never becomes enabled, this leads to the contradiction.
\end{proof}

\begin{lemma}\label{lem:tcf}
Any configuration in $\cal{LC}$ is a $(f\Delta^d,\Delta^d,0,f)$-time contained configuration of the spanning tree construction, where $f$ is the number of Byzantine processes and $d$ is the diameter of the subsystem consisting of all the correct processes.
\end{lemma}

\begin{proof}
Let $\rho_0$ be a configuration of $\cal{LC}$ and $e=\rho_0,\rho_1,\ldots$ be an execution starting from $\rho_0$. First, we show that any $0$-correct process takes at most $\Delta^d$ actions in $e$, where $d$ is the diameter of the subsystem consisting of all the correct processes.

Let $F$ be the set of Byzantine processes in $e$. Consider a subsystem $S'$ consisting of all the correct processes: $S'=(P-F, L')$ where $L'=\{l \in L~|~l \in (P-F) \times (P-F)\}$. We prove by induction on the distance $\delta$ from the root in $S'$ that a correct process $v$ $\delta$ hops away from $r$ in $S'$ executes its action at most $\Delta^\delta$ times in $e$.

\begin{itemize}
\item Induction basis ($\delta=1$):\\
Let $v$ be any correct process neighboring to the root $r$. Since $\rho_0$ is a legitimate configuration, $prnt_r=0$ and $level_r=0$ hold at $\rho_0$ and remain unchanged in $e$. Thus, if $prnt_v=r$ and $level_v=1$ hold in a configuration $\sigma$, then $v$ never changes $prnt_v$ or $level_v$ in any execution starting from $\sigma$. Since $prnt_v=r$ and $level_v=1$ hold within the first $\Delta_v -1\leq \Delta$ actions of $v$, $v$ can execute its action at most $\Delta$ times.

\item Induction step (with induction assumption):\\
Let $v$ be any correct process $\delta$ hops away from the root $r$ in $S'$, and $u$ be a correct neighbor of $v$ that is $\delta-1$ hops away from $r$ in $S'$ (this process exists by the assumption that the subgraph of correct processes of $S$ is connected). From the induction assumption, $u$ can execute its action at most $\Delta^{\delta-1}$ times.

Assume that $prnt_v=u$ and $level_v=level_u+1$ hold in a given configuration $\sigma$. We can observ that $v$ is not enabled until $u$ does not modify its state. Then, the round-robin order used for pointers modification allows us to deduce that $v$ executes at most $\Delta_v\leq \Delta$ actions between two actions of $u$ (or before the first action of $u$). By the induction assumption, $u$ executes its action at most $\Delta^{\delta-1}$ times. Thus, $v$ can execute its action at most $\Delta + \Delta \times (\Delta^{\delta-1}) = \Delta^\delta$ times.
\end{itemize}

Consequently, any $0$-correct process takes at most $\Delta^d$ actions in $e$.

We say that a Byzantine process $b$ deceive a correct neighbor $v$ in the step $\rho\mapsto\rho'$ if the state of $b$ makes the guard of an action of $v$ true in $\rho$ and if $v$ executes this action in this step.

As a $0$-disruption can be caused only by an action of a Byzantine process from a legitimate configuration, we can bound the number of $0$-disruptions by counting the total number of times that correct processes are deceived of neighboring Byzantine processes.

If a $0$-correct $v$ is deceived by a Byzantine neighbor $b$, it takes necessarily $\Delta_v$ actions before being deceiving again by $b$ (recall that we use a round-robin policy for $prnt_v$). As any $0$-correct process $v$ takes at most $\Delta^d$ actions in $e$, $v$ can be deceived by a given Byzantine neighbor at most $\Delta^{d-1}$ times. A Byzantine process can have at most $\Delta$ neighboring correct processes and thus can deceive correct processes at most $ \Delta \times \Delta^{d-1} = \Delta^d$ times.  We have at most $f$ Byzantine processes, so the total number of times that correct processes are deceived by neighboring Byzantine processes is $f\Delta^d$.

Hence, the number of $0$-disruption in $e$ is bounded by $f\Delta^D$. It remains to show that any $0$-disruption have a finite length to prove the result.

By contradiction, assume that there exists an infinite $0$-disruption $d=\rho_i,\ldots$ in $e$. This implies that for all $j\geq i$, $\rho_j$ is not in $\cal{LC}$, which contradicts Lemma \ref{lem:convergence}. Then, the result is proved.
\end{proof}

\begin{theorem}[Strong-stabilization]\label{theorem:ss}
Protocol $ss$-$ST$ is a $(f\Delta^d, 0, f)$-strong stabilizing protocol for the spanning tree construction, where $f$ is the number of Byzantine processes and $d$ is the diameter of the subsystem consisting of all the correct processes.
\end{theorem}

\begin{proof}
From Lemmas \ref{lem:closure} and \ref{lem:tcf}, it is sufficient to show that $ss$-$ST$ eventually reaches a configuration in $\cal{LC}$. Lemma \ref{lem:convergence} allows us to conclude.
\end{proof}

\subsection{Time Complexities}

\begin{proposition}
The $(f\Delta^d,0,f)$-process-disruption time of $ss$-$ST$ is $\Delta^d$ where $d$ is the diameter of the subsystem consisting of all the correct processes.
\end{proposition}

\begin{proof}
This result directly follows from Theorem \ref{theorem:ss} and Lemma \ref{lem:tcf}.
\end{proof}

\begin{proposition}
The $(f\Delta^d,0,f)$-stabilization time of $ss$-$ST$ is $O((n-f)\Delta^d)$ rounds where $f$ is the number of Byzantine processes and $d$ is the diameter of the subsystem consisting of all the correct processes.
\end{proposition}

\begin{proof}
By the construction of the algorithm, any correct process $v$ which has a correct neighbor $u$ takes at most $\Delta$ steps between two actions of $u$.

Given two processes $u$ and $v$, we denote by $d'(u,v)$ the distance between $u$ and $v$ in the subgraph of correct processes of $S$. We are going to prove the following property by induction on $i>0$:

$(P_i)$: any correct process $v$ such that $d'(v,r)=i$ takes at most $2\cdot\overset{i}{\underset{j=1}{\sum}}\Delta^j$ steps in any execution starting from any configuration.

\begin{itemize}
\item Induction basis ($i=1$):\\
Let $v$ be a correct neighbor of the root $r$. By the algorithm, we know that the root $r$ takes at most one step (because $r$ is correct). By the previous remark, we know that $v$ takes at most $\Delta$ steps before and after the action of $r$. Consequently, $v$ takes at most $2\Delta$ steps  in any execution starting from any configuration.
\item Induction step ($i>1$ with induction assumption):\\
Let $v$ be a correct process such that $d'(v,r)=i$. Denote by $u$ one neighbor of $v$ such that $d'(u,r)=i-1$ (this process exists by the assumption that the subgraph of correct processes of $S$ is connected).

By the previous remark, we know that $v$ takes at most $\Delta$ steps before the first action of $u$, between two actions of $u$ and after the last action of $u$. By induction assumption, we know that $u$ takes at most $2\cdot\overset{i-1}{\underset{j=1}{\sum}}\Delta^j$ steps. Consequently, $v$ takes at most $A$ actions where:
\[A=\Delta+\left(2\cdot\overset{i-1}{\underset{j=1}{\sum}}\Delta^j\right)\cdot\Delta+\Delta=2\cdot\overset{i}{\underset{j=1}{\sum}}\Delta^j\]
\end{itemize}
Since there is $(n-f)$ correct processes and any correct process satisfies $d'(v,r)<d$, we can deduce that the system reach a legitimate configuration in at most $O((n-f)\Delta^d)$ steps of correct processes.

As a round counts at least one step of a correct process, we obtain the result.
\end{proof}

\section{Strongly-Stabilizing Tree Orientation}

\subsection{Problem Definition}

In this section, we consider only \emph{tree systems}, \emph{i.e.} distributed systems containing no cycles. We assume that all processes in a tree system are identical and thus no process is distinguished as a root.

Informally, \emph{tree orientation} consists in transforming a tree system (with no root) into a rooted tree system. Each process $v$ has an O-variable $prnt_v$ to designate a neighbor as its parent. Since processes have no identifiers, $prnt_v$ actually stores $k~(\in \{1,~2,\ldots,~\Delta_v\})$ to designate its $k$-th neighbor as its parent. But for simplicity, we use $prnt_v=k$ and $prnt_v=u$ (where $u$ is the $k$-th neighbor of $v$) interchangeably.

The goal of tree orientation is to set $prnt_v$ of every process $v$ to form a rooted tree. However, it is impossible to choose a single process as the root because of impossibility of symmetry breaking. Thus, instead of a single root process, a single \emph{root link} is determined as the root: link $(u, v)$ is the root link when processes $u$ and $v$ designate each other as their parents (Fig.~\ref{fig:TO}(a)). From any process $w$, the root link can be reached by following the neighbors designated by the variables $prnt$.

When a tree system $S$ has a Byzantine process (say $w$), $w$ can prevent communication between subtrees of $S-\{w\}$\footnote{For a process subset $P'~(\subseteq P)$, $S-P'$ denotes a distributed system obtained by removing processes in $P'$ and their incident links.}. Thus, we have to allow each of the subtrees to form a rooted tree independently. We define the specification predicate $spec(v)$ of the tree orientation as follows.

\begin{center}
$spec(v):\forall u~(\in N_v) [(prnt_v=u)\vee (prnt_u=v)\vee (u$ is Byzantine faulty)].
\end{center}

Note that the tree topology, the specification and the uniquiness of $prnt_v$ (for any process $v$) imply that, for any $0$-legitimate configuration, there is at most one root link in any connected component of correct processes. Hence, in a fault-free system, there exists exactly one root link in any $0$-legitimate configuration.
  
Figure~\ref{fig:TO} shows examples of $0$-legitimate configurations (a) with no Byzantine process and (b) with a single Byzantine process $w$. The arrow attached to each process points the neighbor designated as its parent. Notice that, from Fig.~\ref{fig:TO}(b), subtrees consisting of correct processes are classified into two categories: one is the case of forming a rooted tree with a root link in the subtree ($T_1$ in Fig.~\ref{fig:TO}(b)), and the other is the case of forming a rooted tree with a root process, where the root process is a neighbor of a Byzantine process and designates the Byzantine process as its parent ($T_2$ in Fig.~\ref{fig:TO}(b)).

\begin{figure}[t]
 \begin{center}
   \includegraphics{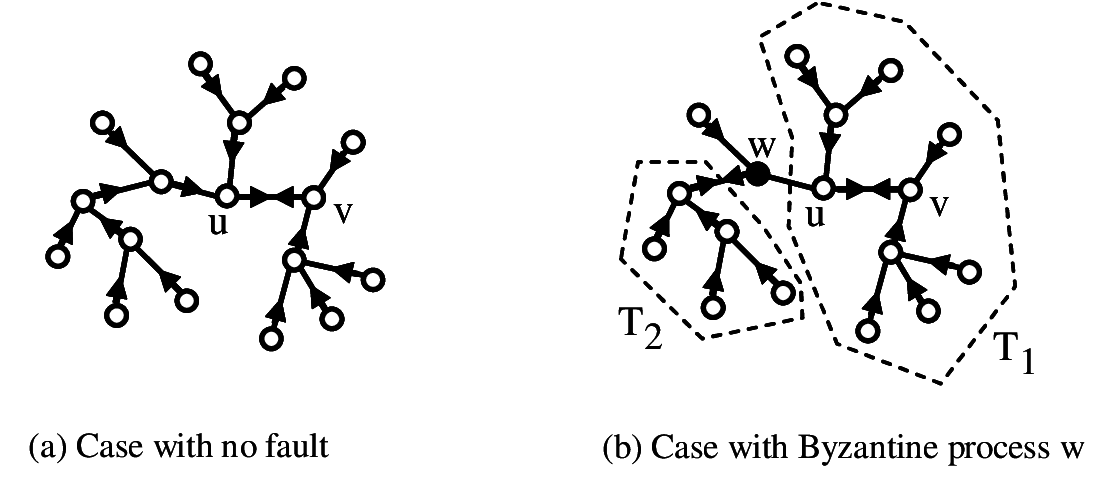}
 \end{center}
 \caption{Tree orientation}
 \label{fig:TO}
\end{figure}

\subsection{Impossibility for Two Byzantine Processes}

Tree orientation seems to be a very simple task. Actually, for tree orientation in fault-free systems, we can design a self-stabilizing protocol that chooses a link incident to a center process\footnote{A process $v$ is a center when $v$ has the minimum eccentricity where eccentricity is the largest distance to a leaf. It is known that a tree has a single center or two neighboring centers.} as the root link: in case that the system has a single center, the center can choose a link incident to it, and in case that the system has two neighboring centers, the link between the centers become the root link. However, tree orientation becomes impossible if we have Byzantine processes. By the impossibility results of \cite{NA02c}, we can show that tree orientation has no $(o(n), 1)$-strictly stabilizing protocol; \emph{i.e.} the Byzantine influence cannot be contained in the sense of ``strict stabilization'', even if only a single Byzantine process is allowed.

An interesting question is whether the Byzantine influence can be contained in a weaker sense of ``strong stabilization''. The following theorem gives a negative answer to the question: if we have two Byzantine processes, bounding the number of disruptions is impossible. We prove the impossibility for more restricted schedules, called the \emph{central daemon}, which disallows two or more processes to make actions at the same time. Notice that impossibility results under the central daemon are stronger than those under the distributed daemon in the sense that impossibility results under the central daemon also hold for the distributed daemon.

\begin{theorem}\label{th:imposs}
Even under the central daemon, there exists no deterministic $(t, o(n), 2)$-strongly stabilizing protocol for tree orientation where $t$ is any (finite) integer and $n$ is the number of processes.
\end{theorem}

\begin{proof}
Let $S=(P,L)$ be a chain (which is a special case of a tree system) of $n$ processes: $P=\{v_1,~v_2,\ldots,v_n\}$ and $L=\{(v_i,~v_{i+1})\ |\ 1 \le i \le n-1\}$.

For purpose of contradiction, assume that there exists a $(t, o(n), 2)$-strongly  stabilizing protocol $A$ for some integer $t$. In the following, we show, for $S$ with Byzantine processes $v_1$ and $v_n$, that $A$ has an execution $e$ containing an infinite number of $o(n)$-disruptions. This contradicts the assumption that $A$ is a $(t, o(n), 2)$-strongly stabilizing protocol. 

In $S$ with Byzantine processes $v_1$ and $v_n$, $A$ eventually reaches a configuration $\rho_1$ that is $o(n)$-legitimate for \emph{spec} and $o(n)$-stable by definition of a $(t, o(n), 2)$-strongly stabilizing protocol. This execution to $\rho_1$ constitutes the prefix of $e$.

To construct $e$ after $\rho_1$, consider another chain $S'=(P',L')$ of $3n$ processes and an execution of $A$ on $S'$, where let $P'=\{u_1,~u_2,\ldots,u_{3n}\}$ and $L'=\{(u_i,~u_{i+1})\ |\ 1 \le i \le 3n-1\}$. We consider the initial configuration $\rho'_1$ of $S'$ that is obtained by concatenating three copies (say $S'_1, S'_2$ and $S'_3$) of $S$ in $\rho_1$ where only the central copy $S'_2$ is reversed right-and-left (Fig.~\ref{fig:imp_proof}). More formally, the state of $w_i$ and of $w_{2n+i}$ in $\rho'_1$ is the same as the one of $v_i$ in $\rho_1$ for any $i\in\{1,\ldots,n\}$. Moreover, for any $i\in\{1,\ldots,n\}$, the state of $w_{n+i}$ in $\rho'_1$ is the same as the one of $v_i$ in $\rho_1$ with the following modification: if $prnt_{v_i}=v_{i-1}$ (respectively $prnt_{v_i}=v_{i+1}$) in $\rho_1$, then $prnt_{w_{n+i}}=w_{n+i+1}$ (respectively $prnt_{w_{n+i}}=w_{n+i-1}$) in $\rho'_1$. For example, if $w$ denotes a center process of $S$ (\emph{i.e.} $w=v_{\lceil n/2 \rceil}$), then $w$ is copied to $w'_1=u_{\lceil n/2 \rceil}, w'_2=u_{2n+1-\lceil n/2 \rceil}$ and $w'_3=u_{2n+\lceil n/2 \rceil}$, but only $prnt_{w'_2}$ designates the neighbor in the different direction from $prnt_{w'_1}$ and $prnt_{w'_3}$. From the configuration $\rho'_1$, protocol $A$ eventually reaches a legitimate configuration $\rho''_1$ of $S'$ when $S'$ has no Byzantine process (since a strongly stabilizimg protocol is self-stabilizig in a fault-free system). In the execution from $\rho'_1$ to $\rho''_1$, at least one $prnt$ variable of $w'_1, w'_2$ and $w'_3$ has to change (otherwise, it is impossible to guarantee the uniquiness of the root link in $\rho''_1$). Assume $w'_i$ changes $prnt_{w'_i}$.

Now, we construct the execution $e$ on $S$ after $\rho_1$. The main idea of this proof is to construct an execution on $S$ indistinguishable (for correct processes) from one of $S'$ because Byzantine processes of $S$ behave as correct processes of $S'$. Since $v_1$ and $v_n$ are Byzantine processes in $S$, $v_1$ and $v_n$ can simulate behavior of the end processes of $S'_i$ (\emph{i.e.} $u_{(i-1)n+1}$ and $u_{in}$). Thus, $S$ can behave in the same way as $S'_i$ does from $\rho'_1$ to $\rho''_1$. Recall that process $w'_1$ modifies its pointer in the execution of $S'_i$ does from $\rho'_1$ to $\rho''_1$. Consequently, we can construct the execution that constitutes the second part of $e$, where $prnt_w$ changes at least once. Letting the resulting configuration be $\rho_2$ (that coincides with the configuration $\rho''_i$ of $S'_i$), $\rho_2$ is clearly $o(n)$-legitimate for \emph{spec} and $o(n)$-stable. Thus, the second part of $e$ contains at least one $o(n)$-disruption.

By repeating the argument, we can construct the execution $e$ of $A$ on $S$ that contains an infinite number of $o(n)$-disruptions.
\end{proof}

\begin{figure}[t]
 \begin{center}
   \includegraphics[scale=0.7]{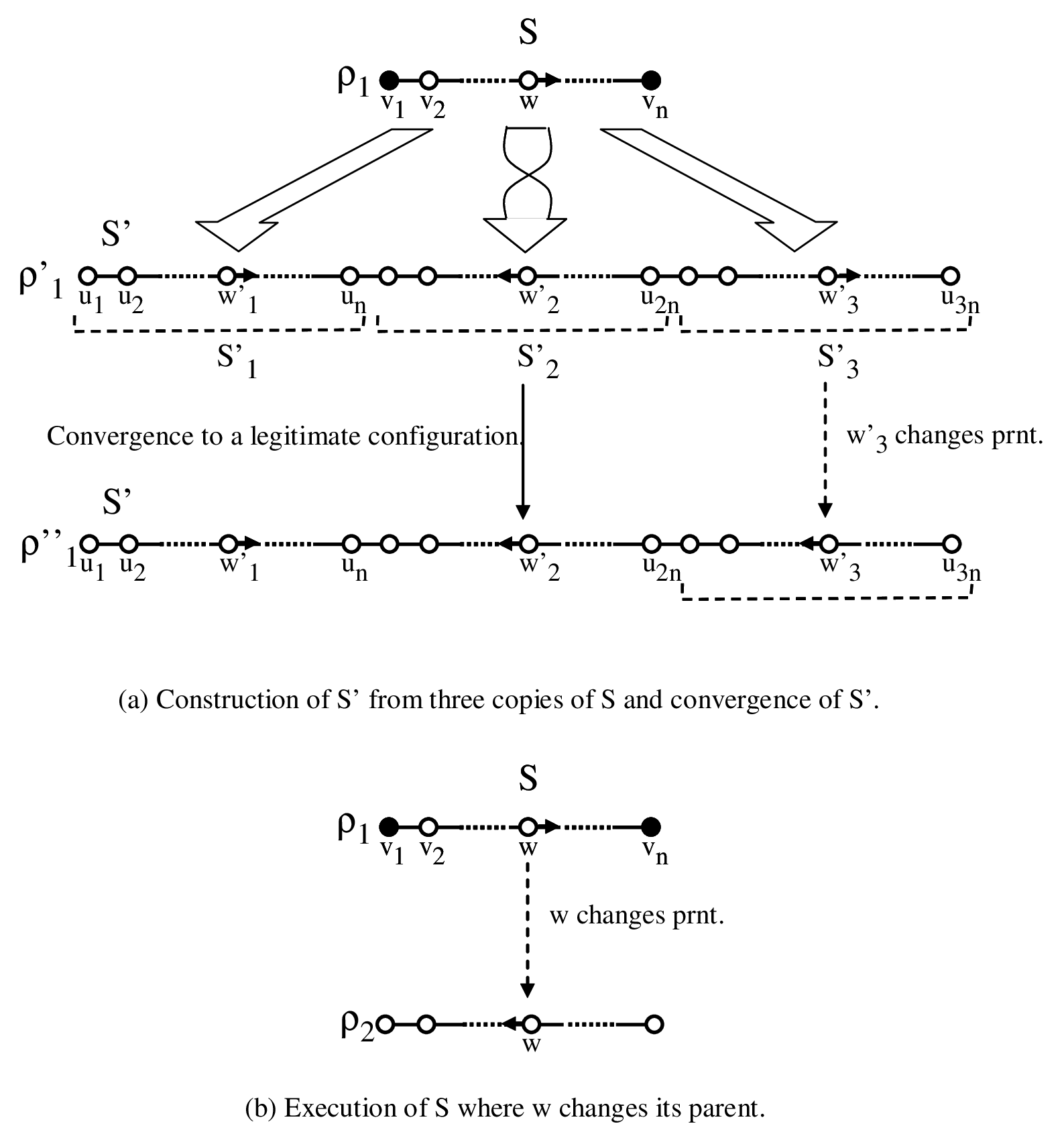}
 \end{center}
 \caption{Construction of execution where $w$ of $S$ changes its parent infinitely often.}
 \label{fig:imp_proof}
\end{figure}

\subsection{A Strongly Stabilizing Protocol for a Single Byzantine Process}

\subsubsection{Protocol $ss$-$TO$}

In the previous subsection, we proved that there is no strongly stabilizing protocol for tree orientation if two Byzantine processes exist. In this subsection, we consider the case with at most a single Byzantine process, and present a $(\Delta,0,1)$-strongly stabilizing tree orientation protocol $ss$-$TO$. Note that we consider the distributed daemon for this possibility result.

In a fault-free tree system, tree orientation can be easily achieved by finding a center process. A simple strategy for finding the center process is that each process $v$ informs each neighbor $u$ of the maximum distance to a leaf from $u$ through $v$. The distances are found and become fixed from smaller ones. When a tree system contains a single Byzantine process, however, this strategy cannot prevent perturbation caused by wrong distances the Byzantine process provides: by reporting longer and shorter distances than the correct one alternatively, the Byzantine process can repeatedly pull the chosen center closer and push it farther.

The key idea of protocol $ss$-$TO$ to circumvent the perturbation is to restrict the Byzantine influence to one-sided effect: the Byzantine process can pull the chosen root link closer but cannot push it farther. This can be achieved using a non-decreasing variable $level_v$ as follows: when a process $v$ finds a neighbor $u$ with a higher level, $u$ chooses $v$ as its parent and copies the level value from $u$. This allows the Byzantine process (say $z$) to make its neighbors choose $z$ as their parents by increasing its own level. However, $z$ can not make neighbor change their parents to other processes by decreasing its own level. Thus, the effect the Byzantine process can make is one-sided.

Protocol $ss$-$TO$ is presented in Fig.~\ref{fig:ssTO}. For simplicity, we regard constant $N_v$ as denoting the neighbors of $v$ and regard variable $prnt_v$ as storing a parent of $v$. Notice that they should be actually implemented using the ordinal numbers of neighbors that $v$ locally assigns.  

The protocol is composed of three rules. The first one ({\tt GA1}) is enabled when a process has a neighbor which provides a strictly greater level. When the rule is executed, the process chooses this neighbor as its parent and computes its new state in function of this neighbor. The rule {\tt GA2} is enabled when a process $v$ has a neighbor $u$ (different from its current parent) with the same level such that $v$ is not the parent of $u$ in the current oriented tree. Then, $v$ chooses $u$ as parent, increments its level by one and refresh its shared registers. The last rule ({\tt GA3}) is enabled for a process when there exists an inconsistence between its variables and its shared registers. The execution of this rule leads the process to compute the consistent values for all its shared registers.

\begin{figure}[h!]
\begin{center}
\begin{tabbing}
xxx \= xxx \= xxx \= xxx \= xxx \= \kill
{\tt constants of process $v$} \\
\> $\Delta_v =$ the degree of $v$; \\
\> $N_v = $ the set of neighbors of $v$;  \\
{\tt variables of process $v$} \\
\> $prnt_v$: a neighbor of $v$; // $prnt_v=u$ if $u$ is a parent of $v$.    \\
\> $level_v$: integer; \\
{\tt variables in shared register $r_{v,u}$} \\
\> $r$-$prnt_{v,u}$: boolean; 
     // $r$-$prnt_{v,u}=$\textit{true} iff $u$ is a parent of $v$.    \\
\> $r$-$level_{v,u}$: integer; // the value of $level_v$ \\
{\tt predicates} \\
\> $pred_1: \exists u \in N_v [r$-$level_{u,v} > level_v]$ \\
\> $pred_2: \exists u \in N_v-\{prnt_v\} [(r$-$level_{u,v} = level_v)\wedge  
   (r$-$prnt_{u,v}=$\textit{false}$)]$ \\
\> $pred_3: ((r$-$prnt_{v,prnt_v},r$-$level_{v,prnt_v})\neq ($\textit{true}$,level_v))\vee$\\
\> \> \> $(\exists u\in N_v-\{prnt_v\},(r$-$prnt_{v,u},r$-$level_{v,u})\neq ($\textit{false}$,level_v))$\\
{\tt atomic actions} // represented in form of guarded actions \\
\> {\tt GA1:}$pred_1\ \longrightarrow$ \\ 
\> \> \> Let $u$ be a neighbor of $v$ s.t. 
         $r$-$level_{u,v}=\max_{w \in N_v}r$-$level_{w,v}$;\\
\> \> \> $prnt_v :=u;\ level_v := r$-$level_{u,v};$\\ 
\> \> \> $(r$-$prnt_{v,u}, r$-$level_{v,u}) := ($\textit{true}$, level_v)$;\\
\> \> \> {\tt for each} $w \in N_v-\{u\}$ {\tt do} 
                 $(r$-$prnt_{v,w}, r$-$level_{v,w}):=($\textit{false}$, level_v)$;\\
\> {\tt GA2:}$\neg pred_1\wedge pred_2\ \longrightarrow$ \\
\> \> \> Let $u$ be a neighbor of $v$ s.t. 
         $(r$-$level_{u,v} = level_v)\wedge (r$-$prnt_{u,v}=$\textit{false}$)$;\\
\> \> \> $prnt_v :=u;\ level_v := level_v+1;$\\
\> \> \> $(r$-$prnt_{v,u}, r$-$level_{v,u}) := ($\textit{true}$, level_v)$;\\
\> \> \> {\tt for each} $w \in N_v-\{u\}$ {\tt do} 
                 $(r$-$prnt_{v,w}, r$-$level_{v,w}):=($\textit{false}$, level_v)$;\\
\> {\tt GA3:}$\neg pred_1\wedge \neg pred_2\wedge pred_3 \longrightarrow$ \\
\> \> \> $(r$-$prnt_{v,prnt_v}, r$-$level_{v,prnt_v}) := ($\textit{true}$, level_v)$;\\
\> \> \> {\tt for each} $w \in N_v-\{prnt_v\}$ {\tt do} 
                 $(r$-$prnt_{v,w}, r$-$level_{v,w}):=($\textit{false}$, level_v)$;
\end{tabbing}
\end{center}
\caption{Protocol $ss$-$TO$ (actions of process $v$)}
\label{fig:ssTO}
\end{figure}

\subsubsection{Closure of Legitimate Configurations of $ss$-$TO$}

We refine legitimate configurations of protocol $ss$-$TO$ into several sets of configurations and show their properties. We cannot make any assumption on the initial values of register variables. But once a correct process $v$ executes its action, variables of its output registers have values consistent with the process variables: $r$-$prnt_{v,prnt_v}=true$, $r$-$prnt_{v,w}=false\ (w \in N_v-\{prnt_v\})$, and $r$-$level_{v,w}=level_v\ (w \in N_v)$ hold. In the following, we assume that all the variables of output registers of every correct process have consistent values.

First we consider the fault-free case.

\begin{definition}[$\mathcal{LC}_0$]\label{def:legit0}
In a fault-free tree, we define the set of configurations $\mathcal{LC}_0$ as the set of configurations such that: (a) $spec(v)$ holds for every process $v$ and (b) $level_u=level_v$ holds for any processes $u$ and $v$.
\end{definition}

In any configuration of $\mathcal{LC}_0$, variables $prnt_v$ of all processes form a rooted tree with a root link as Fig.~\ref{fig:TO}(a), and all variables $level_v$ have the same value.

\begin{lemma}\label{lem:legit0}
In a fault-free tree, once protocol $ss$-$TO$ reaches a configuration $\rho$ in $\mathcal{LC}_0$, it remains at $\rho$. 
\end{lemma}

\begin{proof}
Consider any configuration $\rho$ in $\mathcal{LC}_0$. Since all variables $level_v$ have the same value, the guard of {\tt GA1} cannot be true in $\rho$. Since $spec(v)$ holds at every process in $\rho$, there exist no neighboring processes $u$ and $v$ such that $prnt_u \ne v$ and $prnt_v \ne u$ holds. It follows that the guard of {\tt GA2} cannot be true in $\rho$. Once each process executes an action, all the variables of its output registers are consistent with its local variables, and thus, the guard of {\tt GA3} cannot be true.
\end{proof}

For the case with a single Byzantine process, we define the following sets of configurations.

\begin{definition}[$\mathcal{LC}_1$]\label{def:legit1}
Let $z$ be the single Byzantine process in a tree system. A configuration is in the set $\mathcal{LC}_1$ if every subtree (or a connected component) of $S$-$\{z\}$ satisfies either the following (C1) or (C2).
\begin{enumerate}
\renewcommand{\labelenumi}{(C\arabic{enumi})}
\item
(a) $spec(u)$ holds for every correct process $u$,
(b) $prnt_v=z$ holds for the neighbor $v$ of $z$, and
(c) $level_w \ge level_x$ holds for any neighboring correct processes $w$ and $x$ where $w$ is nearer than $x$ to $z$. 
\item
(d) $spec(u)$ holds for every correct process $u$, and
(e) $level_v=level_w$ holds for any correct processes $v$ and $w$.
\end{enumerate}
\end{definition}

\begin{definition}[$\mathcal{LC}_2$]\label{def:legit2}
Let $z$ be the single Byzantine process in a tree system. A configuration is in the set $\mathcal{LC}_2$ if every subtree (or a connected component) of $S$-$\{z\}$ satisfies the condition (C1) of Definition \ref{def:legit1}.
\end{definition}

In any configuration of $\mathcal{LC}_2$, every subtree forms the rooted tree with the root process neighboring the Byzantine process $z$. For configurations of $\mathcal{LC}_2$, the following lemma holds.

\begin{lemma}\label{lem:st-legit}
Once protocol $ss$-$TO$ reaches a configuration $\rho$ of $\mathcal{LC}_2$, it remains in configurations of $\mathcal{LC}_2$ and, thus, no correct process $u$ changes $prnt_u$ afterward. That is, any configuration of $\mathcal{LC}_2$ is $(0,1)$-contained.
\end{lemma}

\begin{proof}
Consider any execution $e$ starting from a configuration $\rho$ of $\mathcal{LC}_2$. In $\rho$, every subtree of $S-\{z\}$ forms the rooted tree with the root process neighboring the Byzantine process $z$. Note that, as long as no correct process $u$ changes $prnt_u$ in $e$, action {\tt GA2} cannot be executed at any correct process. On the other hand, if a process $u$ executes action {\tt GA1} in  $e$, $level_{prnt_u} \ge level_u$ necessarily holds immediately this action. Consequently, if we assume that no correct process $u$ changes $prnt_u$ in $e$ (by execution of {\tt GA1}) then every configuration of $e$ is in $\mathcal{LC}_2$. To prove the lemma, it remains to show that $e$ contains no activation of {\tt GA1} by a correct process. In the following, we show that any correct process $u$ never changes $prnt_u$ in $e$.

For contradiction, assume that a correct process $u$ changes $prnt_u$ first among all correct processes. Notice that every correct process $v$ can execute {\tt GA1} or {\tt GA3} but cannot change $prnt_v$ before $u$ changes $prnt_u$. Also notice that $u$ changes $prnt_u$ to its neighbor (say $w$) by execution of {\tt GA1} and $w$ is a correct process. From the guard of {\tt GA1}, $level_w > level_u$ holds immediately before $u$ changes $prnt_u$. On the other hand, since $w$ is a correct process, $w$ never changes $prnt_w$ before $u$. This implies that $prnt_w = u$ holds immediately before $u$ changes $prnt_u$, and thus $level_u \ge level_w$ holds. This is a contradiction. 
\end{proof}

Notice that a correct process $u$ may change $level_u$ by execution of {\tt GA1} even after a configuration of $\mathcal{LC}_2$. For example, when the Byzantine process $z$ increments $level_z$ infinitely often, every process $u$ may also increment $level_u$ infinitely often.
 
\begin{lemma}\label{lem:legit1}
Any configuration $\rho$ in $\mathcal{LC}_1$ is $(\Delta_z,1,0,1)$-time contained where $z$ is the Byzantine process.
\end{lemma}

\begin{proof}
Let $\rho$ be a configuration of $\mathcal{LC}_1$. Consider any execution $e$ starting from $\rho$. By the same discussion as the proof of Lemma~\ref{lem:st-legit}, we can show that any subtree satisfying (C1) at $\rho$ always keeps satisfying the condition and no correct process $u$ in the subtree changes $prnt_u$ afterward.

Consider a subtree satisfying (C2) at $\rho$ and let $y$ be the neighbor of the Byzantine process $z$ in the subtree. From the fact that variables $prnt_u$ form a rooted tree with a root link and all variables $level_u$ have the same value in the subtree at $\rho$, no process $u$ in the subtree changes $prnt_u$ or $level_u$ unless $y$ executes $prnt_y:=z$ in $e$. When $prnt_y:=z$ is executed, $level_y$ becomes larger than $level_u$ of any other process $u$ in the subtree. Since the value of variable $level_u$ of each correct process $u$ is non-decreasing, every correct neighbor (say $v$) of $y$ eventually executes $prnt_v:=y$ and $level_v:=level_y$ (by {\tt GA1}). By repeating the argument, we can show that the subtree eventually reaches a configuration satisfying (C1) in $O(d')$ rounds where $d'$ is the diameter of the subtree. It is clear that any configuration before reaching the first configuration satisfying (C1) is not in $\mathcal{LC}_1$, and that each process $u$ changes $prnt_u$ at most once during the execution. 

Therefore, any execution $e$ starting from $\rho$ contains at most $\Delta_z$ 0-disruptions where each correct process $u$ changes $prnt_u$ at most once.
\end{proof}

\subsubsection{Convergence of $ss$-$TO$}

We first show convergence of protocol $ss$-$TO$ to configurations of $\mathcal{LC}_0$ in a \emph{fault-free} case.

\begin{lemma}\label{lem:conv0}
In a fault-free tree system, protocol $ss$-$TO$ eventually reaches a configuration of $\mathcal{LC}_0$ from any initial configuration.
\end{lemma}

\begin{proof}
We prove the convergence to a configuration of $\mathcal{LC}_0$ by induction on the number of processes $n$. It is clear that protocol $ss$-$TO$ reaches a configuration  of $\mathcal{LC}_0$ from any initial configuration in case of $n=2$.

Now assume that protocol $ss$-$TO$ reaches a configuration of $\mathcal{LC}_0$ from any initial configuration in case that the number of processes is $n-1$ (inductive hypothesis), and consider the case that the number of processes is $n$.

Let $u$ be any leaf process and $v$ be its only neighbor and $\rho$ be an arbitrary configuration.
In a first time, we show that any execution $e$ starting from $\rho$ reaches in a finite time a configuration such that $level_v \ge level_u$ holds. If this condition holds in $\rho$, we have the result. Otherwise ($level_v<level_u$), $u$ is continuously enabled by {\tt GA1} (until the condition is true). Hence, the condition becomes true (by an activation of $v$) or this action is executed by $u$ in a finite time. In both cases, we obtain that $level_v \ge level_u$ holds in at most one round.

After that, process $u$ can execute only guarded action {\tt GA1} or {\tt GA3} since $prnt_u=v$ always holds. Thus, after the first round completes, $prnt_u=v$ and $level_v \ge level_u$ always hold (indeed, $v$ can only increase its $level$ variable and $level$ variable of $u$ can only take greater values than $v$'s). It follows that $v$ never executes $prnt_v:=u$ in the second round and later. This implies that $e$ reaches in a finite time a configuration $\rho'$ such that (a) $prnt_v \ne u$ always holds after $\rho'$ , or (b) $prnt_v= u$ always holds after $\rho'$ (since $v$ cannot execute $prnt_v:=u$ after $\rho'$ if $prnt_v \ne u$).

In case (a), the behavior of $v$ after $\rho'$ is never influenced by $u$: $v$ behaves exactly the same even when $u$ does not exist. From the inductive hypothesis, protocol $ss$-$TO$ eventually reaches a configuration $\rho''$ such that $S-\{u\}$ satisfies the condition of $\mathcal{LC}_0$ and remains in $\rho''$ afterward (from Lemma \ref{lem:legit0}). After $u$ executes its action at $\rho''$, $level_u = level_v$ holds and thus the configuration of $S$ is in $\mathcal{LC}_0$.

Now consider case (b), where we do not use the inductive hypothesis. The fact that $prnt_v=u$ (and $prnt_u=v$) always holds after $\rho'$ implies that $level_v$ (and also $level_u$) remains unchanged after $\rho'$. Assume now that a neighbor $w~(\ne u)$ of $v$ satisfies continuously $level_w\neq level_v$ or $prnt_w\neq v$ from a configuration $\rho''$ of $e$ after $\rho'$. If $w$ satisfies continuously $level_w> level_v$ from $\rho''$, then $v$ executes {\tt GA1} in a finite time, this is a contradiction. If $w$ satisfies continuously $level_w< level_v$ from $\rho''$, then $w$ executes {\tt GA1} in a finite time and takes a $level$ value such that $level_w\geq level_v$, that contradicts the fact that $w$ satisfies continuously $level_w< level_v$ from $\rho''$. This implies that $level_w= level_v$ and $prnt_w= v$ in a finite time in any execution starting from $\rho'$. As $v$ does not modify its state after $\rho'$, $w$ is never enabled after $\rho'$. This implies that the fragment of $S$ consisting of processes within distance two from $u$ reaches a configuration satisfying the condition of $\mathcal{LC}_0$ and remains unchanged. We can now apply the same reasoning by induction on the distance of any process to $u$ and show that $ss$-$TO$ eventually reaches a configuration in $\mathcal{LC}_0$ where link $(u, v)$ is the root link.

Consequently, protocol $ss$-$TO$ reaches a configuration of $\mathcal{LC}_0$ from any initial configuration. 
\end{proof}

Now, we consider the case with a single Byzantine process.

\begin{lemma}\label{lem:conv1}
In a tree system with a single Byzantine process, protocol $ss$-$TO$ eventually reaches a configuration of $\mathcal{LC}_1$ from any initial configuration.
\end{lemma}

\begin{proof}
Let $z$ be the Byzantine process, $S'$ be any subtree (or a connected component) of $S-\{z\}$ and $y$ be the process in $S'$ neighboring $z$ (in $S$).

We prove, by induction on the number of processes $n'$ of $S'$, that $S'$ eventually reaches a configuration satisfying the condition (C1) or (C2) of Definition \ref{def:legit1}.

It is clear that $S'$ reaches a configuration satisfying (C1) from any initial configuration in case of $n'=1$.

Now assume that $S'$ reaches a configuration satisfying (C1) or (C2) from any initial configuration in case of $n'=k-1$  (inductive hypothesis), and consider the case of $n'=k\ (\ge 2)$.

From $n' \ge 2$, there exists a leaf process $u$ in $S'$ that is not neighboring the Byzantine process $z$. Let $v$ be the neighbor of $u$. Since processes $u$ and $v$ are correct processes, we can show the following by the same argument as the fault-free case (Lemma \ref{lem:conv0}): after some configuration $\rho$, (a) $prnt_v \ne u$ always holds, or (b) $prnt_v= u$ always holds. In case (a), we can show from the inductive hypothesis that $S'$ eventually reaches a configuration satisfying (C1) or (C2). In case (b), we can show that $S'$ eventually reaches a configuration satisfying (C2) where link $(u, v)$ is the root link.

Consequently, protocol $ss$-$TO$ reaches a configuration of $\mathcal{LC}_1$ from any initial configuration.
\end{proof}

The following main theorem is obtained from Lemmas \ref{lem:legit0}, \ref{lem:st-legit}, \ref{lem:legit1}, \ref{lem:conv0} and \ref{lem:conv1}.

\begin{theorem}\label{th:strong}
Protocol $ss$-$TO$ is a $(\Delta,0,1)$-strongly stabilizing tree-orientation protocol.
\end{theorem}

\subsubsection{Round Complexity of $ss$-$TO$}

In this subsection, we focus on the round complexity of $ss$-$TO$. First, we show the following lemma.

\begin{lemma}\label{lem:roundp}
Let $v$ and $u$ be any neighbors of $S$. Let $S'$ be the subtree of $S-\{v\}$ containing $u$ and $h(v,u)$ be the largest distance from $v$ to a leaf process of $S'$. If $S' \cup \{v\}$ contains no Byzantine process, $prnt_v := u$ of {\tt GA1} or {\tt GA2} can be executed only in the first $2 h(v,u)$ rounds. Moreover, in round $2 h(v,u)$+1 or later, $level_v$ remains unchanged as long as $prnt_v = u$ holds.
\end{lemma}

\begin{proof}
We prove the lemma by induction on $h(v,u)$.

First consider the case of $h(v,u)=1$, where $u$ is a leaf process. When the first round completes, all the output registers of every process becomes consistent with the process variables. Since $u$ is a leaf process, $prnt_u = v$ always holds. It follows that process $v$ can execute $prnt_v := u$ only in {\tt GA1}. Once $v$ executes its action in the second round, $level_v \ge level_u$ holds and $prnt_v := u$ of {\tt GA1} cannot be executed afterward (see proof of Lemma \ref{lem:conv0}). Thus, $prnt_v := u$ of {\tt GA1} can be executed only in the first and second rounds. It is clear that in round $3$ or later, $level_v$ remains unchanged as long as $prnt_v = u$ holds.

We assume that the lemma holds when $h(v,u) \le k-1$ (inductive hypothesis) and consider the case of $h(v,u) = k$. We assume that $prnt_v := u$ of {\tt GA1} or {\tt GA2} is executed in round $r$, and show that $r \le 2k$ holds in the following. Variable $level_v$ is also incremented in the action, and let $\ell$ be the resultant value of $level_v$. In the following, we consider two cases.

\begin{itemize}
\item Case that $prnt_v := u$ of {\tt GA1} is executed in round $r$: when $prnt_v := u$ is executed, $level_u=\ell$ holds. But $level_u<\ell$ holds when $v$ executes its action in round $r-1$; otherwise, $v$ reaches a state with $level_v \ge \ell$ in round $r-1$ and cannot execute $prnt_v := u$ (with $level_v := \ell$) in round $r$. This implies that $u$ incremented $level_u$ to $\ell$ in round $r-1$ or $r$.

In the case that $u$ makes the increment of $level_u$ by {\tt GA1}, $u$ executes $prnt_u :=w$ for $w\ (\ne v)$ in the same action. Since $h(u,w)<h(v,u)$ holds, the action is executed in the first $2 h(u,w)$ rounds from the inductive hypothesis. Consequently, $prnt_v := u$ of {\tt GA1} is executed in round $2 h(u,w)+1\ (< 2 h(v,u))$ at latest.

In the case that $u$ makes the increment of $level_u$ by {\tt GA2}, $u$ executes $prnt_u :=w$ for some $w\ (\in N_u)$ in the same action, where $w=v$ may hold. For the case of $w \ne v$, we can show, by the similar argument to the above, that $prnt_v := u$ is executed in round $2 h(u,w)+1\ (< 2 h(v,u))$ at latest. Now consider the case of $w=v$. Then $level_v = level_u = \ell -1$, $prnt_v \ne u$ and $prnt_u \ne v$ hold immediately before $u$ executes $prnt_u:=v$ and $level_u:=\ell$. Between the actions of $level_u := \ell -1$ (with $prnt_u := w\ (w \ne v)$) and $level_u := \ell$ (with $prnt_u := v$), $v$ can execute its action at most once; otherwise, $level_v \ge \ell-1$ holds after the first action, and $level_v \ge \ell$ or $prnt_v = u$ holds after the second action. This implies that $level_u := \ell -1$ with $prnt_u := w\ (w \ne v)$ is executed in the previous or the same round as the action of $level_u := \ell$, and thus, in round $r-2$ or later. Since $h(u,w)<h(v,u)$ holds, the action is executed in the first $2 h(u,w)$ rounds from the inductive hypothesis. Consequently, $prnt_v := u$ of {\tt GA1} is executed in round $2 h(u,w)+2\ (\le 2 h(v,u))$ at latest.

\item Case that $prnt_v := u$ is executed in {\tt GA2}: then $level_v = level_u = \ell -1$, $prnt_v \ne u$ and $prnt_u \ne v$ hold immediately before $v$ executes $prnt_v:=u$ and $level_v:=\ell$. Between the executions of $level_v := \ell -1$ and $level_v := \ell$, $u$ can execute its action at most once, and $u$ executes $prnt_u := w$ for some $w\ (\ne v)$ in the action.Since $h(u,w)<h(v,u)$ holds, this action is executed in the first $2 h(u,w)$ rounds from the inductive hypothesis. Consequently, $prnt_v := u$ is executed in round $2 h(u,w)+1\ (< 2 h(v,u))$. 
\end{itemize}

It remains to show that $level_v$ remains unchanged in round $2 h(v,u)$+1 or later, as long as $prnt_v = u$ holds. Now assume that $prnt_v=u$ holds at the end of round $2 h(v,u)$.

\begin{itemize}
\item Case that $prnt_u=v$ holds at the end of round $2 h(v,u)$: since $h(u,w) < h(v,u)$ for any $w \in N_u-\{v\}$, $prnt_u := w$ cannot be executed in round $2 h(v,u)+1$ or later from the inductive hypothesis, and so $prnt_u=v$ holds afterward. Thus, it is clear that $level_v$ remains unchanged as long as $prnt_v=u$ (and $prnt_u=v$) holds.

\item Case that $prnt_u \ne v$ holds at the end of round $2 h(v,u)$: let $prnt_u = w$ hold for some $w \in N_u-\{v\}$ at the end of round $2 h(v,u)$. Since $h(u,w) < h(v,u)$, $level_u$ remains unchanged as long as $prnt_u = w$ holds from the inductive hypothesis. It follows that $level_v$ remains unchanged as long as $prnt_v=u$ and $prnt_u = w$ hold. Since $h(u,x) < h(v,u)$ for any $x \in N_u-\{v\}$, $prnt_u := x$ cannot be executed in round $2 h(v,u)+1$ or later, but $prnt_u := v$ can be executed. Immediately after execution of $prnt_u := v$, $level_v=level_u$ holds if $prnt_v$ remains unchanged. Thus, it is clear that $level_v$ remains unchanged as long as  $prnt_v=u$  (and $prnt_u=v$) holds.
\end{itemize}
\end{proof}

The following lemma holds for the fault-free case.

\begin{lemma}\label{lem:round0}
In a fault-free tree system, protocol $ss$-$TO$ reaches a configuration of $\mathcal{LC}_0$ from any initial configuration in $O(d)$ rounds where $d$ is the diameter of the tree system $S$. 
\end{lemma}

\begin{proof}
Lemma \ref{lem:roundp} implies that, after round $2d+1$ or later, no process $v$ changes $prnt_v$ or $level_v$ and thus the configuration remains unchanged. Lemma \ref{lem:conv0} guarantees that the final configuration is a configuration in $\mathcal{LC}_0$.
\end{proof}

For the single-Byzantine case, the following lemma holds.

\begin{lemma}\label{lem:round1}
In a tree system with a single Byzantine process, protocol $ss$-$TO$ reaches a configuration of $\mathcal{LC}_1$ from any initial configuration in $O(n)$ rounds.
\end{lemma}

\begin{proof}
Let $z$ be the Byzantine process and $S'$ be any subtree of $S -\{z\}$. Let $v$ be the neighbor of $z$ in $S'$. From Lemma \ref{lem:roundp}, $v$ cannot execute $prnt_v := w$ for any $w \in N_v-\{z\}$ in round $2 d'+1$ or later, where $d'$ is the diameter of $S'$. We consider the following two cases depending on $prnt_v$.

\begin{itemize}
\item Case 1: there exists $w \in N_v-\{z\}$ such that $prnt_v =w$ at the end of round $2d'$ and $prnt_v$ remains unchanged during the following $d'$ rounds (from round $2d'+1$ to round $3d'$).

From Lemma \ref{lem:roundp}, $level_v$ also remains unchanged during the $d'$ rounds. By the similar discussion to that in proof of Lemma \ref{lem:roundp}, we can show that $S'$ reaches a configuration satisfying the condition (C2) of Definition \ref{def:legit1} by the end of round $3d'$.

\item Case 2: $prnt_v=z$ at the end of round $2d'$ or there exists at least one configuration during the following $d'$ rounds (from round $2d'+1$ to round $3d'$) such that $prnt_v=z$ holds.

Let $c$ be the configuration where $prnt_v=z$ holds. From Lemma \ref{lem:roundp}, $prnt_v=z$ always holds after $c$. We can show, by induction of $k$ that, a fraction of $S'$  consisting of processes with distance up to $k$ from $v$ satisfies the condition (C1) at the end of $k$ rounds after $c$. Thus, $S'$ reaches a configuration satisfying the condition (C1) of Definition \ref{def:legit1} by the end of round $4d'$.
\end{itemize}

After a subtree reaches a configuration satisfying the condition (C2), its configuration may change into one satisfying the condition (C1) and the configuration may not satisfy (C1) or (C2) during the transition. However, Lemma \ref{lem:legit1} guarantees that the length of the period during the subtree does not satisfy (C1) or (C2) is $O(d')$ rounds, where $d'$ is the diameter of the subtree. Since the total of diameters of all the subtrees in $S-\{z\}$ is $O(n)$, the convergence to a configuration of $\mathcal{LC}_1$ satisfying (C1) or (C2) can be delayed at most $O(n)$ rounds.
\end{proof}

Finally, we can show the following theorem.

\begin{theorem}\label{th:round}
Protocol $ss$-$TO$ is a $(\Delta,0,1)$-strongly stabilizing tree-orientation protocol. The protocol reaches a configuration of $\mathcal{LC}_0 \cup \mathcal{LC}_1$ from any initial configuration. The protocol may move from a legitimate configuration to an illegitimate one because of the influence of the Byzantine process, but it can stay in illegitimate configurations during the total of $O(n)$ rounds (that are not necessarily consecutive) in the whole execution.   
\end{theorem}

\begin{proof}
Theorem \ref{th:strong} shows that $ss$-$TO$ is a $(\Delta,0,1)$-strongly stabilizing tree-orientation protocol. Lemma \ref{lem:round0} and \ref{lem:round1} guarantee that $ss$-$TO$ reaches a configuration of $\mathcal{LC}_0 \cup \mathcal{LC}_1$ from any initial configuration within $O(n)$ rounds. For the case with a single Byzantine process (say $z$), each subtree of $S-\{z\}$ may experience an illegitimate period (not satisfying the condition (C1) or (C2)) after such a configuration. However, Lemma \ref{lem:legit1} guarantees that the length of the illegitimate period is $O(d')$ where $d'$ is the diameter of the subtree. Since the total of diameters of all the subtrees in $S-\{z\}$ is $O(n)$, the total length of the periods that does not satisfy (C1) or (C2) is $O(n)$ rounds.
\end{proof}

\section{Concluding Remarks}

We introduced the notion of strong stabilization, a property that permits self-stabilizing protocols to contain Byzantine behaviors for tasks where strict stabilization is impossible. In strong stabilization, only the first Byzantine actions that are performed by a Byzantine process may disturb the system. If the Byzantine node does not execute Byzantine actions, but only correct actions, its existence remains unnoticed by the correct processes. So, by behaving properly, the Byzantine node may have the system disturbed arbitrarily far in the execution. By contrast, if the Byzantine node executes many Byzantine actions at the beginning of the execution, there exists a time after which those Byzantine actions have no impact on the system. As a result, the faster an attacker spends its Byzantine actions, the faster the system become resilient to subsequent Byzantine actions. An interesting trade-off appears: the more actually Byzantine actions are performed, the faster the stabilization of our protocols is (since the number of steps performed by correct processes in response to Byzantine disruption is independent from the number of Byzantine actions). Our work raises several important open questions:
\begin{enumerate}
\item is there a trade-off between the number of perturbations Byzantine nodes can cause and the containment radius ? In this paper, we strove to obtain optimal containment radius in strong stabilization, but it is likely that some problems do not allow strong stabilization with containment radius 0. It is then important to characterize the difference in containment radius when the task to be solved is ``harder'' than tree orientation or tree construction. 
\item is there a trade-off between the total number of perturbations Byzantine nodes can cause and the number of Byzantine nodes, that is, is a single Byzantine node more effective to harm the system than a team of Byzantine nodes, considering the same total number of Byzantine actions ? A first step in this direction was recently taken by~\cite{YMB10c}, where Byzantine actions are assumed to be upper bounded, for the (global) problem of leader election. Their result hints that only Byzantine actions are relevant, independently of the number of processes that perform them. It is thus interesting to see if the result still holds in the case of potentially infinite number of Byzantine actions.
\end{enumerate}

\singlespacing

\bibliography{biblio}
\bibliographystyle{plain}

\end{document}